\documentclass[11pt,a4paper]{article}
\usepackage{fullpage,xspace,times}
\usepackage{amsmath,upgreek}
\usepackage{amsfonts}
\usepackage{amssymb}
\newtheorem{theorem}{Theorem}

\newtheorem{claim}[theorem]{Claim}

\newtheorem{proposition}[theorem]{Proposition}

\newcommand{\Xomit}[1]{ }
\newenvironment{proof}[1][Proof]{\textbf{#1.} }{\ \rule{0.5em}{0.5em}}

\mathchardef\mhyphen="2D

\newcommand{\eps}{\upvarepsilon}
\newcommand{\veps}{\varepsilon}
\newcommand{\del}{\updelta}

\begin{document}

\title{Parallel solutions for preemptive  makespan scheduling \\ on two identical machines}
\date{}
\author{Leah Epstein\thanks{
Department of Mathematics, University of Haifa, Haifa, Israel.
\texttt{lea@math.haifa.ac.il}. }}
\maketitle

\begin{abstract}
We consider online preemptive scheduling of jobs arriving one by one, to be assigned to two identical machines, with the goal of makespan minimization.
We study the effect of selecting the best solution out of two independent solutions constructed in parallel in an online fashion.
Two cases are analyzed, where one case is purely online, and in the other one jobs are presented sorted by non-increasing sizes.
We show that using two solutions rather than one improves the performance significantly, but that an optimal solution cannot be obtained for any constant number of solutions constructed in parallel. Our algorithms have the best possible competitive ratios out of algorithms for each one of the classes.
\end{abstract}

\section{Introduction}
Online scheduling problems on multiprocessors, where jobs are scheduled one by one on $m \geq 2$ machines, are defined as follows. A sequence of jobs is presented, where each job $j$ has a size $p_j>0$, and it should be assigned to run on the machines.
An online algorithm creates a single solution, by assigning one job at a time, without any knowledge of future jobs, and without a possibility to modify the assignment of previously presented jobs. The algorithm does not know when the input will stop, and it is judged by the quality of its constructed solution after the input sequence is terminated.
Since real-life scheduling problems allow some flexibility, there are also relaxed models, where reassignment is sometimes possible (see for example \cite{SSS04,ChenLBDH11,EL11,AH17_2,Levin22}). Another relaxed approach is the design of semi-online algorithms, where some prior knowledge on the input is available (see for example \cite{AH12,KK13,KKG15}). In this work, we consider a variant where the online algorithm can construct not a single solution, but a constant number of solutions, where the algorithm can choose one of its solutions as its output when the input is stopped \cite{KKST97}. The solution is selected after the input ends as the current best solution.

To specify the definition of an assignment of a job to machines, we distinguish several models. In all models, for every solution and at each time in the schedule, every machine can run at most one job, and any job can be run on at most one machine.
In the non-preemptive variant, each job is assigned to one of the machines.  Every job $j$ receives a time interval of the form $[S_j,C_j)$, where $S_j$ is the starting time of $j$ and $C_j$ is the completion time of $j$. In this variant, it can be assumed that the time intervals on each machine (for the different jobs assigned to this machine) will be consecutive. The completion time or load of a machine is the supremum time that it is running a job, and the makespan is the maximum load of any machine. When multiple solutions are considered, we sometimes use the term {\it maximum completion time} for the makespan of a single solution (as opposed to the makespan which is the smallest maximum completion time over all solutions).
For the machine model called identical machines,  the processing times of jobs are equal to their sizes, and the completion time or load of a machine is the total size of jobs assigned to it. In particular it holds that $C_j=S_j+p_j$ for every job $j$. Uniformly related machines have speeds, such that the speed of machine $i$ is denoted by $s_i\geq 0$, and $C_j=S_j+\frac{p_j}{s_i}$ for every job $j$ assigned to machine $i$, that is, the processing time of job $j$  on machine $i$ is $\frac{p_j}{s_i}$, and the completion time of machine $i$ is the total size of jobs assigned to it divided by its speed.

Here, we study the preemptive variant. In this problem, every job has to be assigned completely to machines, possibly by splitting the job into parts, such that it is allocated to non-overlapping time intervals, on one machine or on multiple machines. The makespan is again the supremum time that any machine runs a job (i.e., the completion time or load of a machine is defined as the supremum time it runs a job, and the makespan is defined as before). A complete assignment for identical machines means that the total length of time intervals for every job $j$ is $p_j$. That is, the job is split into parts that are scheduled independently but not in parallel. For uniformly related machines, letting $t_{ij} \geq 0$ denote the total length of time intervals of job $j$ on machines $i$, it holds that $\sum_{i=1}^m t_{ij}\cdot s_i=p_j$. That is, the job can still be split into parts assigned to disjoint time slots, and the time slot needed for each part is based on the speed of the machine running it.

We analyze online (and semi-online) algorithms via the competitive ratio. This is the worst-case ratio between the cost (makespan in our case) of the online algorithm and the optimal offline cost (for the same input). In the model where parallel solutions are considered, the cost of the algorithm is the minimum cost out of the costs of its solutions. All the stated problems have constant competitive algorithms, which are algorithms whose competitive ratios are independent of the input. One can find results for uniformly related machines and identical machines, in the non-preemptive case and the preemptive case, for unsorted and sorted inputs, in a number of articles \cite{ChVlWo95,Sg97,ENSSW01,EJS09,ES09a,ES11,Gr66,Gr69,SSW00,EpsteinF02,EpsteinF05,Albers99,FaKeTu89,AAFPW97,
BCK00,Fleischer2000,JSSB13,Kovacs10}, and here we specify the known results for identical machines and the preemptive variant, which we study in this work.

For online algorithms and a single solution, the best possible competitive ratio for $m$ machines is $\frac{m^m}{m^m-(m-1)^m}<\frac{e}{e-1}\approx 1.5819767$, if the input is unsorted \cite{ChVlWo95,Sg97}, so the best possible competitive ratio for $m=2$ is $\frac 43$. For inputs where jobs arrive sorted by non-increasing sizes, the tight result is known for every $m$, its supremum value is approximately $1.36603$, and its value for $m=2$ is $1.2$ \cite{SSW00}. This last article by Seiden, Sgall, and Woeginger \cite{SSW00} deals with several variants of sorted inputs for identical machines (see also \cite{CKK}), while sorted inputs for uniformly related machines were studied in various scenarios as well \cite{EpsteinF02,Kovacs10}.

The model with several solutions constructed in an online fashion is strongly related to online problems with advice. In such models, the algorithm has access to {\it advice bits}, which are hints telling the algorithm how to proceed \cite{BKKKM17,BFKLM17,BFKM18}. The two models (with advice and with multiple solutions) are essentially equivalent. The last claim holds since the advice bits can be guessed or enumerated by parallel solutions whose number is $2$ to the power of advice bits. On the other hand, given an algorithm with parallel solutions, one can ask for the advice of the index of the solution that will be the best for the planned input.
However, work on online algorithms with advice is usually involved with algorithms where the number of advice bits is fairly large, and the number of solutions is an integer power of $2$, where this large number of advice bits comes from an unknown source. For example, if the algorithm is told (by the advice) which machine to use for every job, the problem becomes uninteresting.  Even if the number of advice bits is independent of the input size, still the number of bits can be a function of a parameter $\eps$, and this affects the performance of the algorithm, in the sense that it performs better for very small values of $\eps$. The resulting number of solutions for a desired value of $\eps$ is far from being sufficiently small in practice, typically even for values of $\eps$ that are not very small (for example, for $\eps=\frac 12$).
Here, we focus on a slightly different line of research, where the capacity of advice is very small, and it consists of just a few bits \cite{FGK019}. In our main results, the advice in fact consists of a single bit, that is, the number of solutions is just two.

We believe that the first problem that was ever studied with respect to multiple solutions is non-preemptive online scheduling of jobs arriving one by one to minimize makespan for the case $m=2$ \cite{KKST97}. In that work, an algorithm of competitive ratio $
\frac 43$ is designed for the case of two solutions, and a matching lower bound on the competitive ratio is provided. This can be compared to the best possible ratio of $\frac 32$ with a single solution \cite{Gr66,FaKeTu89}. A different algorithm with a larger number of advice bits was designed later \cite{Doh15}. Other algorithms, including ones for arbitrary numbers of machines, with a constant (but large) number of advice bits, and with a non-constant numbers of advice bits (but very small competitive ratios), were given \cite{AH17_1,RRVS,Doh15,BFKM18}. A notable result by Albers and Hellwig \cite{AH17_1} is that for obtaining a competitive ratio below $\frac 43$, the number of parallel solutions has to be at least linear in the number of machines, $m$. In the same work, an algorithm with a constant (but very large) number of parallel solutions, achieving an upper bound close to this threshold of $\frac 43$ on the competitive ratio for identical machines, was presented. Other models of scheduling \cite{BFKM18,BKKKM17} were analyzed as well, with respect to a large amount of advice.

There is no direct connection between the competitive ratios for the preemptive and non-preemptive problem with the same input. For large numbers of machines, the best possible competitive ratio for the variant with preemption is much smaller \cite{ChVlWo95,RudinC03}. This is also the case for two machines and arbitrary inputs \cite{ChVlWo95,FaKeTu89}, but not for sorted inputs. For inputs consisting of jobs presented sorted by non-increasing sizes, the competitive ratio for preemptive algorithms is $1.2$ \cite{SSW00}, and it is $\frac 76\approx 1.166667$ for non-preemptive algorithms \cite{SSW00}. Moreover, for identical jobs (each of size $1$), the non-preemptive case is trivial (via a round robin assignment), while the preemptive case is harder. Note that for two machines and two solutions, one can obtain an optimal solution for the preemptive case. Specifically, one solution will be the non-preemptive optimal solution, while the other one assigns the first three jobs carefully (the first job is assigned during the time $[0,1)$ on the first machine, the second job is assigned during $[1,1.5)$ on the first machine and during $[0,0.5)$ on the second machine, and the third job is assigned during $[0.5,1.5)$ on the second machine), and remaining jobs are assigned non-preemptively via round robin. The first solution is optimal for the preemptive problem for inputs with even numbers of jobs (and for an input with a single job), and the second solution is optimal for inputs with odd numbers of jobs. Interestingly, for a single solution and two machines, the case of identical jobs is known to be the hardest \cite{SSW00}, while this does not hold for the case of two solutions, which is implied by our results.

\noindent{\bf Our results.} In this work, we study the case of two identical machines with respect to preemptive online solutions. We start with a brief discussion of constant number of solutions, where the constant may be large. Specifically, we show that no algorithm with a fixed number of solutions can be optimal (in the sense that its competitive ratio is $1$), which holds even for a sorted input, but it is possible to obtain a competitive ratio arbitrarily close to $1$, by allowing relatively large (but fixed) numbers of solutions. Our main results are algorithms whose competitive ratios are the best possible, both with two solutions. The first one is for the case where job sizes are arbitrary, and the second one is for the case where jobs are presented in a sorted order, such that they are sorted by non-increasing size.
In both cases, the competitive ratios are significantly smaller than those of the case with a single solution. Specifically, for the general case, the competitive ratio is approximately $1.236068$ while the best possible competitive ratio with a single solution is $\frac 43$ \cite{ChVlWo95}. For non-increasing sizes, the competitive ratio is $1.10102051$ while the best possible competitive ratio with a single solution is $1.2$ \cite{SSW00}.

The approach for obtaining our algorithms is based on a careful design of each one of the two solutions, such that their properties complement each other.
This is combined with ideas used in previous work. The algorithmic approach of \cite{KKST97} for the design of a non-preemptive algorithm is as follows. The algorithm has two solutions, where one has a good performance, while the other one accumulates jobs on one machine. Once a large job arrives, the two solutions swap roles. Thus, the method is that the two solutions have different goals, where one is very balanced and one is very imbalanced, such that it requires a very large job to become balanced. The algorithmic idea is to exchange the roles of the two solutions under some condition on the input.
Here, we cannot normally just assign all jobs to one machine, yet one solution usually will have a larger makespan than the other one. The solutions are constructed such that the ratio between machine completion times is fixed for each solution, unless there was a large job. The solutions change roles upon the arrival of a large job, since the balance changes, and the way each solution is balanced is swapped. This approach of an attempt to keep a fixed ratio between machine completion times is a method that is often used for preemptive algorithms \cite{ChVlWo95,Eps01}. The algorithm computes the optimal offline makespan and tries to reach a makespan of the competitive ratio times this optimal offline makespan on its most loaded machine. Thus, it sets a threshold in each step. In our algorithm, we use two different thresholds for the two solutions, such that swapping the solutions also swaps the thresholds.

We summarize the most relevant previous work and our results in Table \ref{tabbb}.
{\begin{table}[h!] {\renewcommand{\arraystretch}{1.2   }
\small
$$
\begin{array}{||c||c||c||}
\hline
\mbox{scheduling model} & \mbox{single solution} & \mbox{two solutions} \\
\hline
\hline
m \mbox{ \ machines, unsorted }  & 1.581977 \ \ \cite{ChVlWo95} & - \\ \hline
m \mbox{ \ machines, sorted } &
 1.36603 \ \ \cite{SSW00} &  - \\ \hline
\mbox{  two machines, unsorted }     & 1.333333 \ \ \cite{ChVlWo95} & 1.236068 \\ \hline
\mbox{two  machines, sorted } &1.2  \ \ \cite{SSW00} & 1.101021 \\ \hline

\hline
\hline
\end{array}
$$}
\caption{\label{tabbb} A summary of known and new results for the competitive ratio of online preemptive scheudling on identical machines. The stated values are approximate, but all results are tight.
Results without a citation are proved here in Sections \ref{unso},\ref{sor}. Empty entries mean that there is no improved result for this case.}
\end{table}
}

Given a set of solutions constructed by an online algorithm, we say that the makespan (for the output of this algorithm) is the minimal maximum completion time, where the minimum is taken over solutions, and the maximum is taken over machines. In the case of two solutions $A_1$ and $A_2$, where the final machine loads are $a_{1,1}$ and $a_{1,2}$ for $A_1$ and they are $a_{2,1}$ and $a_{2,2}$ for $A_2$, the makespan is $\min_{k=1,2} \{\max_{i=1,2} a_{k,i}\}$. In the case where no idle time is introduced, we have $a_{1,1}+a_{1,2}=a_{2,1}+a_{2,2}=W$ for the total size of jobs, $W$.

\section{Constant numbers of solutions}
In this section, we provide observations regarding large numbers of solutions. The goal is to analyze the status of the problem with respect to its ``advice complexity'' for a large number of advice bits. Note that computing a large number of solutions is not practical, but the negative part of the result could be seen as motivation for using a very small number of solutions or advice bits. The result of the current section is shown for $m = 2$ identical machines, since two machines is the main topic of the article.
Note that for any number of identical machines, and even for any number of uniformly related machines, if the optimal makespan is known, one can obtain an optimal offline solution in an online fashion \cite{McN,Ep03,ES09a}. For the case of identical machines, one can use the algorithm of McNaughton \cite{McN}, which is very simple. Given $T=\max\{p^{\max},\frac{W}m\}$, where $p^{\max}$ is the maximum job size (and $W$ is the total size of all jobs, as defined earlier), it allocates the time interval $[0,T)$ on every machine, and assigns jobs one by one to consecutive time slots, moving to the next machine after the current machine is fully occupied during the allocated time internal, as long as not all jobs are assigned, and possibly cutting a job into two parts in the case that it was assigned partially to a machine before moving to the next machine. We will use this property (that $T$ is the optimal offline makespan) below.
For non-preemptive scheduling, the knowledge of the optimal offline cost does not allow an online algorithm to compute an optimal offline solution even for two identical machines \cite{AzarR01}.

\begin{proposition}\label{propo1}
For any fixed number of parallel solutions, $M \geq 2$, any online algorithm has competitive ratio $1+\Omega(\frac 1M)$, even if input jobs arrive sorted by non-increasing sizes.
\end{proposition}

The proof is based on the next idea. The number of solutions is finite, but by considering an infinite number of possibilities for the third job (which arrives after two identical jobs), the number of possible offline costs (for balanced optimal solutions) is infinite, since the problem is preemptive. Ordering the solutions by a feature of their action on the first two jobs, we find a gap in the solution set with respect to the possible maximum completion time, and show that no solution will have a maximum completion time that is sufficiently close to the optimal makespan.

\medskip

\noindent\begin{proof}
We start with the unsorted case, and define an input based on the action of all $M$ solutions. The first two jobs have sizes of $M+1$ each. For each solution out of the $M$ solutions, consider the total length of time intervals during which both machines are assigned to run a job or a part of a job  (that is, times when both machines are occupied simultaneously). Let these $M$ values be $b_1,b_2\ldots,b_M$, where $0 \leq b_i \leq M+1$ for $1\leq i \leq M$, where $b_i$ corresponds to the $i$th solution. We use the notation $b_0=0$ and $b_{M+1}=M+1$, and assume without loss of generality that the values are sorted, i.e., $$0=b_0\leq b_1\leq b_2 \leq \cdots \leq b_M \leq b_{M+1}=M+1 \ . $$ By the pigeonhole principle, and since $b_{M+1}-b_0=M+1$, there is a value $j'$ where $0\leq j'\leq M$ such that $b_{j'+1}-b_{j'} \geq 1$.
We use a parameter $\alpha$, where $0<\alpha \leq 1$, and find the largest $j \leq M$ such that $b_{j+1}-b_{j} \geq \alpha$, which has to exist due to the previous claim.

Let $p=2M+2-2\cdot b_{j+1}+\alpha$ be the size of the third job, which is presented next. We have $b_{j+1} \geq b_j+\alpha \geq \alpha$ and $$ 2\cdot b_j +\alpha \leq b_j+b_{j+1} \leq 2\cdot b_{j+1}-\alpha \ . $$
Thus, by $b_{j+1} \leq M+1$, we get  $p \geq 2M+2-2(M+1)+\alpha>0$ and $$p \leq 2(M+1)-(b_j+b_{j+1}) \leq 2(M+1)-2\cdot b_j-\alpha <2(M+1) \ . $$

The current total size of input jobs is $2(M+1)+p$, and thus the optimal offline makespan is $\max\{M+1,p,M+1+\frac p2 \}=M+1+\frac p2$. For solution $i$, the maximum completion time is at least $\max\{b_i+p,2M+2-b_i\}$, since the third job cannot be assigned during the time when two machines are active, and since the total time when exactly one machine is running a job is $2M+2-2\cdot b_i$ (since the time that both machines are running jobs is $b_i$). In fact, $2M+2-2\cdot b_i$ is a lower bound on the maximum completion time for solution $i$ even before the assignment of the job of size $p$.
For $i\geq j+1$, we have $b_i\geq b_{j+1}$, and the maximum completion time for this solution is at least $b_i+p \geq b_{j+1}+p = 2M+2-b_{j+1}+\alpha$.
For $i\leq j$, we have $b_i\leq b_{j}$, and $2M +2- b_i \geq 2M+2 -b_j \geq 2M+2-b_{j+1}+\alpha$ as well. Thus, the best solution out of the $M$ solutions has at least this maximum completion time, and this is a lower bound on the makespan.

The optimal offline cost is $$M+1+\frac p2 = M+1+(M+1-b_{j+1}+\frac{\alpha}2) \ .$$
The ratio between the costs is at least $$\frac{2M+2-b_{j+1}+\alpha}{2M+2-b_{j+1}+\frac{\alpha}2}=1+\frac{\alpha/2}{2M+2-b_{j+1}+\alpha/2} \ .$$ By using $\alpha=1$, and $b_{j+1} \geq \alpha=1$,  we get a ratio of at least $1+\frac{1}{4M+3}$.

We can slightly adapt the proof and show that a similar property holds even if jobs are presented sorted in non-increasing order. The input still consists of two jobs of size $M+1$ and a job of size $p$. This time we use $\alpha=\frac 12$ for the proof, and we show that as a result it will hold that $p<M+1$, and the input is sorted.
If $b_{j+1} \geq \frac{M+2}2$, we have $2\cdot b_{j+1} -\alpha > M+1$, and thus $p < M+1$, so the same calculation is used, and the resulting ratio (which is found by substituting the new value of $\alpha$ and the bound on $b_{j+1}$) is at least $1+\frac{1}{6M+5}$.

To complete the proof, we show that $b_{j+1} \geq \frac{M+2}2$ holds for any $j\in\{0,1,\ldots,M\}$.
For $j=M$, we have $b_{j+1}=M+1 >\frac{M+2}2$ for any $M \geq 2$, and therefore we consider the case $0 \leq j\leq M-1$. Assume by contradiction that $b_{j+1} < \frac{M+2}2$. By the choice of $j$, we have $b_{j'+1}-b_{j'} < \frac 12$ for $j'\geq j+1$. Thus, by $j<M$, the next sum is non-empty, and the inequality is strict.
Specifically, we have $$\sum_{j'=j+1}^{M} (b_{j'+1}-b_{j'}) < \frac{M-j}2  \ . $$ By rearranging the left hand side we get $$\sum_{j'=j+1}^{M} (b_{j'+1}-b_{j'})= b_{M+1}-b_{j+1} =M+1-b_{j+1} > M+1-\frac{M+2}2=\frac{M}{2} \ . $$ Thus, $\frac{M}{2}  < \frac{M-j}2$, a contradiction since $j \geq 0$.
\end{proof}

We now design an algorithm of approximation ratio $1+\del$ for any $0<\del\leq 1$ such that $\frac{1}{\del}$ is an integer.
The idea is to keep a set of solutions with very different allowed maximum completion times, but such that the set of their costs is dense. When a solution becomes too small in this sense compared to the current optimal makespan, its allowed maximum completion time is increased by a large factor, keeping the structure of the set of solutions. When the allowed cost is increased, a new schedule is started. Every solution may have older schedules as well, but those old schedules contribute very little to the machine completion times.

Let $\eps=\frac{\del}3$. The algorithm uses $\frac{1}{\eps^2}=\frac{9}{\del^2}$ solutions, where $\frac{1}{\eps}=\frac{3}{\del}$ is an integer as well.
The algorithm has an invariant that it always has at least one solution whose maximum completion time is at most $1+\del$ times the makespan of an optimal offline solution, that is, at most $1+\del$ times the maximum between the current largest job and the current average machine load. Every solution has a length $L$ associated with it, and its maximum completion time will be at most $(1+\eps)\cdot L$. More specifically, this solution will be composed of a current sub-solution and possibly previous sub-solutions, where the previous sub-solutions require time intervals of at most $\eps\cdot L$ in total, and the current sub-solution requires a time interval of at most $L$. The sub-solutions are concatenated, such that when a new sub-solution is added, the schedule it is started at a certain time which is the end of the previous schedule, and there is no overlap in time with previous sub-solutions. Every current sub-solution is found by applying the algorithm of McNaughton \cite{McN}, that is, jobs are assigned consecutively to the first machine in the time interval of the sub-solution, and the schedule continues on the second machine during the same time interval. This can be done in the case where the optimal makespan does not exceed $L$.

When the first job is presented, schedules with lengths $p_1\cdot (1+\eps)^i$ are created for $i=0,1,\ldots,\frac{1}{\eps^2}-1$. At this time, the schedule of length $p_1$ is optimal, since the optimal offline makespan at that time, $OPT_1$, is equal to $p_1$. For any new job $j$, the following is done. First, the new optimal offline makespan $OPT_j$ is computed. For every schedule whose current length $L$ is smaller than $OPT_j$, a new current length is defined as $L\cdot (1+\eps)^{1/\eps^2}$. The set of lengths will always correspond to $p_1$ multiplied by powers of $1+\eps$, such that these powers are consecutive integers. If at some time the powers are $i_1,i_1+1,\ldots, \frac{1}{\eps^2}+i_1-1$, and $p_1(1+\eps)^{i_2-1}<OPT_j$ while $p_1(1+\eps)^{i_2} \geq OPT_j$ (for some $i_2>i_1$ such that $i_2 \leq \frac{1}{\eps^2}+i_1-1$), the modification is applied, so that the powers become $i_2,i_2+1,\ldots, \frac{1}{\eps^2}+i_2-1$, since for the powers $i_1,i_1+1,\ldots i_2-1$, the length of the current sub-solution increases by a multiplicative factor of $(1+\eps)^{1/\eps^2}$.
The other solution lengths are unchanged, and in particular, if $OPT_j \leq p_1(1+\eps)^{i_1}$, all the lengths of current solutions are unchanged. If all lengths are smaller than $OPT_j$, all powers of $1+\eps$ are increased by an additive term of $\frac{1}{\eps^2}$. Every time that a length of a solution is increased, a new sub-solution is started.  This process may be applied multiple times before job $j$ is assigned (intuitively, if $j$ is very large, this will happen a large number of times), which may happen if $p_1(1+\eps)^{ \frac{1}{\eps^2}+i_1}<OPT_j$, and in the analysis we allow empty subsets of jobs to be assigned to sub-solutions, if necessary.

For any solution, taking all sub-solutions into account, if its current length is $L$, its maximum completion time is at most $$L\cdot (1+(1+\eps)^{-1/\eps^2}+(1+\eps)^{1/\eps^4}+\ldots) < \frac{L}{1-(1+\eps)^{-1/\eps^2}} \leq L\cdot(1+\eps) \ ,$$ where the last inequality is explained next. For any positive integer $k$ and $\eps>0$, we have $(1+\eps)^k>1+k\cdot \eps$, and therefore $(1+\eps)^{1/\eps^2} > 1+\frac 1{\eps}$ which implies $1-(1+\eps)^{-1/\eps^2} > 1-\frac 1{1+1/\eps}=\frac{1}{1+\eps}$.

The property for the makespan (of the best solution) holds by induction as follows. Recall that after the arrival of the first job, we have $OPT_1=p_1$, and thus there is a solution of length $p_1$, whose maximum completion time is exactly $p_1$. As long as the set of solutions is unchanged, this solution is still optimal. Once there is a change, the solution with the smallest length $L_s$ (in the new set of solutions) satisfies $\frac{L_s}{1+\eps} < OPT_{j_1}$ for the job $j_1$ that caused the modification. This holds since the smaller power of $1+\eps$ was too small, and the first power that was not too small was kept. Consider a job $j_2$, where $j_2 \geq j_1$, that arrived after the modification for $j_1$ and before any other modification. We have $\frac{L_s}{1+\eps} < OPT_{j_1} \leq OPT_{j_2}$. The maximum completion time for this solution is at most $L_s\cdot (1+\eps)$, while the optimal offline makespan is at least $\frac{L_s}{1+\eps}$. Thus, its approximation ratio is at most $$(1+\eps)^2 =(1+\del/3)^2=1+2\cdot \frac \del 3 +\frac{\del^2}3 \leq 1+\del \ .$$

\medskip

We summarize with the next theorem, where the lower bound is based on Proposition \ref{propo1}, and the upper bound follows from the last algorithm and its analysis.
\begin{theorem}
For any $\veps>0$, there is an algorithm of competitive ratio at most $1+\veps$ that uses $O(\frac{1}{\veps^2})$ solutions. Any algorithm of competitive ratio at most $1+\veps$ requires $\Omega(\frac{1}{\veps})$ solutions.
\end{theorem}

Note that this algorithm can be easily generalized for $m$ identical machines by applying NcNaughton's algorithm \cite{McN} for multiple identical machines, and it can be generalized to the case of uniformly related machines by using the algorithm of Ebenlendr and Sgall \cite{ES09a}. The only required property is the existence of an optimal semi-online algorithm (with competitive ratio $1$) for the case where the optimal offline makespan is given.

\section{Two identical machines}\label{unso}
In this section and the next section we provide our main results, which are algorithms for two machines with two solutions, for the preemptive variants which are the case of general inputs and the case of non-increasing sizes. We start with the case of inputs with arbitrary job sizes.

Let $\alpha$ be the positive solution of $\alpha^2+\alpha-1=0$, that is, $\alpha=\frac{\sqrt{5}-1}2 \approx 0.618$, where $\alpha=\varphi-1$ for $\varphi=\frac{\sqrt{5}+1}2 \approx 1.618$.
Let $R=2\cdot \alpha =\sqrt{5}-1\approx 1.236$ (and $R^2 \approx 1.527864$), where $R^2+2R-4=0$ (so we have $R+2=\frac 4R$). By these definitions, we also have $R=\frac{2}{\varphi}=2(\varphi-1)$. We will use some straightforward properties of $\varphi$ in the calculations, and in particular the properties $\varphi^2=\varphi+1$, $\varphi-1=\frac{1}{\varphi}$, and $2\varphi+1=\varphi^3$.

We start with a lower bound on the competitive ratio.

\begin{proposition}\label{propo2}
The competitive ratio of any algorithm for $m=2$ identical machines is at least $R$.
\end{proposition}
\begin{proof}
Assume by contradiction that there is an algorithm whose competitive ratio is $r<R$.
We will use three inputs, where the second input extends the first one, and the third input extends the second one. Obviously, the competitive ratio has to be satisfied for all three inputs.

The inputs start with very small jobs (sand) of total size $1$. For this input, an optimal offline solution assigns jobs of total size $\frac 12$ to each machine, for a makespan of $\frac 12$. Consider two solutions, $T$ and $X$, constructed by an online algorithm, and let $t_1$, $t_2$, $x_1$, and $x_2$ denote the following total lengths of time intervals for the two solutions. The values $t_1$ and $x_1$ are the total lengths of intervals where at least one machine is active, that is, assigned to run a job, for $T$ and $X$, respectively. We assume without loss of generality that $t_1 \leq x_1$. The values $t_2$ and $x_2$ are the total lengths of intervals where two machines are active, for $T$ and $X$, respectively. Thus, the times where exactly one machine is active are $t_1-t_2$ and $x_1-x_2$, for $T$ and $X$, respectively. The makespan is at least $t_1$, since the maximum completion time for $T$ is at least $t_1$, and for $X$, the maximum completion time is at least $x_1\geq t_1$.
Since the total size of jobs is $1$, and the total size of jobs assigned to run is $2\cdot t_2+(t_1-t_2)=t_1+t_2$ for the first solution, and similarly, $x_1+x_2$ for the second solution, we have $t_1+t_2=x_1+x_2=1$.

Since the makespan of the algorithm for the first input, which is based on the first solution, is at least $t_1$, we have $t_1\leq r \cdot \frac 12 = \frac r2$. We find that $t_2=1-t_1\geq 1-\frac r2$. The input continues with a job $j_1$ of size $\alpha$, which is possibly followed with a job $j_2$ of size $1+\alpha=\varphi$. The input with two jobs is the second input, and the input with three jobs is the third input.
The total size of jobs after the arrival of $j_1$ is $1+\alpha=\varphi$, and after the arrival of $j_2$ the total size is $2\varphi$.
The optimal offline makespan after $j_1$ is presented (that is, for the second input) is $\frac{1+\alpha}2$, since   $\frac{1+\alpha}2 > \alpha$, and the optimal offline makespan after $j_2$ is presented is $1+\alpha$. For the first solution, the maximum completion time is at least $t_2+\alpha$ after $j_1$ is presented, and it is at least $t_2+\alpha+1$ after the second job $j_2$ is presented, since there is a total time interval of $t_2$ where no additional job can be assigned after the sand jobs were assigned.
We have $t_2+\alpha\geq 1-\frac{r}2+\alpha > 1-\frac R2+\alpha=1$ (by $t_2\geq 1-\frac r2$ and $\alpha=\frac{R}2$), and by $t_2+\alpha>1$ we also have $t_2+\alpha+1 > 2$.

Assume first that the makespan is achieved for the first solution $T$ for some input. In this case, the competitive ratio for the two inputs above (for the sand together with one job, and for the sand together with two jobs) exceeds $$\frac{2}{1+\alpha}=\frac{2}{1+\frac R2}=\frac{4}{2+R}=R>r \ , $$ by the properties of $R$. Thus, since the competitive ratio may not exceed $r$, the makespan is achieved for the other solution $X$ for both inputs (the input ending with $j_1$, and the input ending with $j_2$). It is left to consider $X$.

Let $u_1$ and $u_2$ be defined for solution $X$, analogously to the previous definitions, after $j_1$ is presented (so $u_1+u_2=1+\alpha$). The makespan after $j_1$ is presented is at least $u_1$, and we get that $u_1 \leq r \cdot \frac{1+\alpha}2$. The makespan for this solution after $j_2$ is presented (that is, for the third input) is at least $u_2+\alpha+1$, since both machines of $X$ are already occupied during a total time of $u_2$ before the assignment of $j_2$, and thus using the competitive ratio for the third input we have $u_2+\alpha+1 \leq r \cdot (1+\alpha)$. Taking the sum of the two inequalities gives $u_1+u_2+\alpha+1 \leq r \cdot \frac 32 \cdot (1+\alpha)$. That is, by $u_1+u_2=1+\alpha$, we get $2(1+\alpha) \leq r \cdot \frac 32 \cdot (1+\alpha)$, which implies $r \geq \frac 43$, a contradiction since $r<R<\frac 43$.
\end{proof}

\medskip
Next, we design an algorithm for two identical machines with two solutions. The algorithm does not introduce any idle time, which we show in the proof.
There are three cases. In the first case, the new job is relatively small, so it is added to each one of the solutions and the properties are kept, and the suitable balance of each solution (a fixed ratio between the two loads) is kept if the solution was balanced according to the required ratio between the two machine loads, or it becomes more balanced. In the second case, the job has intermediate size. It is assigned such that the solutions are swapped and become balanced. In the third case, the job is so large that first a part of it is assigned as if it has intermediate size, and then the residue is added, possibly making the machines of the solutions imbalanced.

We use the following notation. The algorithm will have two solutions denoted by $A$ and $B$, where $A$ will be seen as the better solution, and the maximum completion time of $A$ will be analyzed towards the competitive ratio, so we use the property that the makespan is at most the maximum completion time for $A$. The loads for $A$ after job $j$ was assigned are denoted by $a_1^j$ and $a_2^j$, where $a_2^j \leq a_1^j$ will hold for the two machines. For $B$, the loads are $b_1^j$ and $b_2^j $, where $b_2^j \leq b_1^j$ will hold. Specifically, the loads $a_1^j$ and $b_1^j$ are always the loads for the first machine, and the loads $a_2^j$ and $b_2^j$ are always the loads of the second machine. The makespan after job $j$ is scheduled will be at most $a_1^j$ (this holds by definition regardless of the value $b_1^j$, but we will enforce the property $b_1^j > a_1^j$).
The property of ordered loads of the two machines in both solutions will be implied by the invariants, which will be proved via induction.

Next, we discuss the parameters. As we already mentioned several times, the algorithm tries to maintain a certain ratio between machine loads for each of the two solutions. This is a ratio of $\varphi$ for $A$ and of $2\varphi$ for $B$. If the input contains a very large job whose size is larger than all previous jobs in total,  the ratio may be larger. In the analysis, we will introduce invariants for $A$ and $B$, so that we can ensure that the two solutions always have the required properties.

When a new job that is relatively large is assigned, solutions $A$ and $B$ will be swapped, that is, $A$ will become $B$ and $B$ will become $A$ after the assignment. This does not require any actual action, but later assignments will be based on this swap. In other cases, when the new job is smaller, the roles of $A$ and $B$ will not be swapped. We have $a_1^0=a_2^0=b_1^0=b_2^0$.
We also use the notation $W_j$ for the total size of the first $j$ jobs (where $W_0=0$). Since the algorithm will not introduce idle time, it will always be the case that $a_1^j+a_2^j=W_j$ and $b_1^j+b_2^j=W_j$. As we discussed earlier, the optimal offline cost (or makespan) after $j$ jobs have arrived is $OPT_j=\max\{\frac{W_j}2, p_{\max}^j\}$, where $p_{\max}^j=\max_{j'=1,2,\ldots,j} p_{j'}$.

We will provide the analysis of each one of the cases after all cases are presented. In each case we ensure not only that the properties hold after the assignment, but also that the assignment is valid.

\noindent{\bf The algorithm.}

\smallskip

\noindent{\bf Case 1.} In this case we assume that $p_j \leq W_j \cdot (2-\varphi)$, where $2-\varphi=\frac{1}{\varphi^2}\approx 0.381966$.
In this case the solutions are not swapped. We describe the assignment for each one of the solutions. The algorithm tries to assign the job such that the ratio between the machine completion times is kept, if the ratio held precisely before the arrival of $j$. The ratio will never be closer to $1$, but it may be larger than the required ratio. If the ratio was violated due to an earlier large job, the algorithm tries to correct it completely or partially, in terms of obtaining the ratio exactly, or obtaining a ratio that is closer to the required ratio. The correction is partial, if the job is too small to perform a complete correction.

\medskip

{\bf Case 1.1 for solution $\boldsymbol{A}$.} In this case we assume that we additionally have $a_2^{j-1}+p_j \leq \frac{W_j}{\varphi^2}$.
Job $j$ is assigned to the second machine completely during $[a_2^{j-1},a_2^{j-1}+p_j)$.

\medskip

{\bf Case 1.2 for solution $\boldsymbol{A}$.}
If case 1.1 was not applied, job $j$ is assigned to the second machine during $[a_2^{j-1},\frac{W_j}{\varphi^2})$, and to the first machine during $[a_1^{j-1},\frac{W_j}{\varphi})$.

\medskip

{\bf Case 1.1 for solution $\boldsymbol{B}$.} In this case we assume that we additionally have $b_2^{j-1}+p_j \leq \frac{W_j}{\varphi^3}$.
Job $j$ is assigned to the second machine during $[b_2^{j-1},b_2^{j-1}+p_j)$.

\medskip

{\bf Case 1.2 for solution $\boldsymbol{B}$.}
If case 1.1 was not applied, job $j$ is assigned to the second machine during $[b_2^{j-1},\frac{W_j}{\varphi^3})$, and to the first machine during $[b_1^{j-1},\frac{2W_j}{\varphi^2})$.

\medskip

\noindent{\bf Case 2.} In this case $W_j \cdot (2-\varphi) < p_j \leq \frac{W_j}2$.
The solutions will be swapped, and the ratio will be attained precisely.

For solution $A$ that will become solution $B$, job $j$ is assigned to the second machine during $[a_2^{j-1},\frac{W_j}{\varphi^3})$, and to the first machine during $[a_1^{j-1},\frac{2W_j}{\varphi^2})$.
For solution $B$ that will become solution $A$, job $j$ is assigned to the second machine during $[b_2^{j-1},\frac{W_j}{\varphi^2})$, and to the first machine during $[b_1^{j-1},\frac{W_j}{\varphi})$.

\medskip

\noindent{\bf Case 3.} In this case $p_j > \frac{W_j}2$ (or equivalently $p_j>W_{j-1}$, where the equivalence is shown in the proof, that is, $j$ is larger than all previously arrived jobs together).
The solutions will be swapped in this situation as well.

The assignment is done as follows. First, a fake job of size $W_{j-1}$ is assigned using case 2. The assignment is just a partial assignment of $j$, and we will need to assign another part of size $p_j-W_{j-1}$ of $j$, and this part is assigned to the first machine, during the time interval $$[\frac{2W_{j-1}}{\varphi},\frac{2W_{j-1}}{\varphi}+p_j-W_{j-1})$$ for $A$ and during the time interval $$[\frac{4W_{j-1}}{\varphi^2},\frac{4W_{j-1}}{\varphi^2}+p_j-W_{j-1})   \mbox{ \  for \  } B . $$

This completes the definition of the algorithm. One can see that all possibilities for the size of the new job were considered and the action was defined for all cases. Note that the first job is assigned by the same assignment rules as other jobs. Specifically, it is assigned by case 3, where case 2 is invoked in the empty sense since $W_0=0$, and it is assigned during the time interval $(0,p_j]$ on the first machine for each one of the solutions. This is obviously optimal and does not require a proof, though this case is covered by the proof of case 3.

\noindent{\bf Analysis.}

As mentioned above, we use a collection of properties. Those properties are defined as follows. For solution $A$, we will require that \begin{equation} \frac{W_j}{\varphi} \leq a_1^j \leq  R \cdot OPT_j  \label{n1} \ .\end{equation} Note that in the case where $OPT_j=\frac{W_j}2$ (that is, the case where $W_j \geq 2 \cdot p_{\max}^j$ holds), we will have $ R \cdot OPT_j = \frac2{\varphi} \cdot \frac{W_j}2= \frac{W_j}{\varphi}$. Thus, in this case it will hold that $a_1^j=\frac{W_j}{\varphi}$. Given the invariants and the lack of idle time, we will always have $$a_2^j = W_j-a_1^j \leq W_j-\frac{W_j}{\varphi}=W_j\cdot(1-\frac{1}{\varphi})=\frac{W_j}{\varphi^2} \ , $$ and $a_2^j=\frac{W_j}{\varphi^2}$ holds in the case where $OPT_j=\frac{W_j}2$. Additionally, we will always have $a_1^j>a_2^j$ for $j>0$ (that is, if $W_j>0$), since for $W_j>0$ it holds that $a_1^j \geq \frac{W_j}{\varphi} > \frac{W_j}{\varphi^2} \geq  a_2^j$. Note that $\frac 1{\varphi}\approx 0.618$ and  $\frac 1{\varphi^2}\approx 0.382$.

For solution $B$, we will require that \begin{equation} \frac{2W_j}{\varphi^2} \leq b_1^j \leq  R^2 \cdot OPT_j  \ . \label{n2} \end{equation} In the case where $OPT_j=\frac{W_j}2$, we will have $ R^2 \cdot OPT_j =  \frac{4}{\varphi^2}  \cdot \frac{W_j}2= W_j \cdot \frac{2}{\varphi^2}$. Thus, in this case it will hold that $b_1^j=\frac{2W_j}{\varphi^2}$. We will also have $$b_2^j \leq W_j-\frac{2W_j}{\varphi^2}=W_j\cdot(1-\frac{2}{\varphi^2})=\frac{W_j}{\varphi^3} \ , $$ and $b_2^j=\frac{W_j}{\varphi^3}$ holds in the case where $OPT_j=\frac{W_j}2$. Additionally, we will always have $b_1^j>b_2^j$ for $j>0$, since for $W_j>0$ it holds that $b_1^j \geq \frac{2W_j}{\varphi^2} > \frac{W_j}{\varphi^3} \geq  b_2^j$. Note that $\frac 2{\varphi^2}\approx 0.764$ and  $\frac 1{\varphi^3}\approx 0.236$.

The two required properties, \eqref{n1} and \eqref{n2}, obviously hold before any job is assigned (since all loads are equal to zero and $W_0=OPT_0=0$), and therefore we will prove them using induction, where the base case is $j=0$.

\medskip

We will analyze the algorithm now, and start with case 1.
Note that in case 1 we have $$W_{j-1} =W_j-p_j \geq W_j\cdot(1-(2-\varphi))=(\varphi-1)\cdot W_j$$ and $$W_{j-1} =W_j-p_j \geq p_j(1/(2-\varphi)-1)=\varphi \cdot p_j \ , $$ or equivalently, we have $W_j \leq \varphi \cdot W_{j-1}$ and $p_j \leq \frac{W_{j-1}}{\varphi}$.

\begin{claim}
The assignment of job $j$ in case 1 is valid in the sense that a total size of $p_j$ is assigned, and there is no overlap between the time slots allocated for $p_j$ on the two machines. No idle time is introduced, and the invariants hold after job $j$ is assigned.
\end{claim}
\begin{proof}
For case 1.1 and both solutions, the assignment cannot be invalid since the job is not preempted. No idle time is introduced since the machine that receives the job is the second machine, and the interval allocated for the job starts at time $a_2^{j-1}$ for the first solution and $b_2^{j-1}$ for the second solution. The length of the time interval is equal to the size of the assigned job. In fact, for case 1.2, while the job is preempted and split into two parts, each part is assigned without introducing any idle time for each one of the solutions.

In order to prove all properties excluding the invariants, it is left to analyze the validity of assignment in the remaining cases, and to show that the job is fully assigned in these cases. We show that the two parts of $j$ do not overlap for both solutions, the parts have non-negative sizes, and their total size is equal to $j$.

We start with the assignment for solution $A$.
We show that in case 1.2  the intervals have positive lengths. We have $a_2^{j-1}\leq\frac{W_{j-1}}{\varphi^2}<\frac{W_{j}}{\varphi^2}$, by the invariants for the previous time $j-1$, since $p_j>0$. Additionally,  since case 1.1 not applied, it holds that $a_2^{j-1}+p_j > \frac{W_j}{\varphi^2}$, which implies $$a_1^{j-1}=W_j-p_j-a_2^{j-1} <W_j-\frac{W_j}{\varphi^2}=\frac{W_j}{\varphi} \ ,$$ and thus $a_1^{j-1} \leq \frac{W_{j}}{\varphi}$ holds.
Moreover, $$(\frac{W_j}{\varphi^2}-a_2^{j-1})+(\frac{W_j}{\varphi}-a_1^{j-1})=W_j-W_{j-1}=p_j \ , $$ and $j$ is assigned completely. There is no overlap between the parts since $\frac{W_j}{\varphi^2}\leq a_1^{j-1}$ holds, due to  $\frac{W_j}{\varphi^2} \leq \frac{W_{j-1}}{\varphi}\leq a_1^{j-1}$, by
$W_j \leq \varphi \cdot W_{j-1}$, and by
the invariant for time $j-1$ and the condition of case 1.

Consider now the assignment for solution $B$.
For case 1.2,  the intervals have positive lengths since $b_2^{j-1}\leq\frac{W_{j-1}}{\varphi^3}<\frac{W_{j}}{\varphi^3}$, by the invariants for the previous time $j-1$, and since $p_j>0$. Additionally, $b_1^{j-1} \leq \frac{2W_{j}}{\varphi}$ holds since case 1.1 not applied and thus $b_2^{j-1}+p_j > \frac{W_j}{\varphi^3}$, which implies $$b_1^{j-1}=W_j-p_j-b_2^{j-1} <W_j-\frac{W_j}{\varphi^3}=\frac{2W_j}{\varphi^2} \ .$$ Moreover, $$(\frac{W_j}{\varphi^3}-b_2^{j-1})+(\frac{2W_j}{\varphi^2}-b_1^{j-1})=W_j-W_{j-1}=p_j \ . $$ There is no overlap between the parts  since $\frac{W_j}{\varphi^3}\leq b_1^{j-1}$ holds, due to  $$\frac{W_j}{\varphi^3} \leq \frac{W_{j-1}}{\varphi^2} <  b_1^{j-1} \ , $$ by the invariant for time $j-1$ and the condition of case 1.

We now show that the invariants will hold in all cases after the assignment. For case 1.2, if it is applied for solution $A$, the exact value of $a_1^j$ is $\frac{W_{j}}{\varphi}$, where $R\cdot OPT_j \geq \frac{W_{j}}{\varphi}$. Similarly, if it is applied for solution $B$, the exact value of $b_1^j$ is $\frac{2W_{j}}{\varphi^2}$, where $R^2\cdot OPT_j \geq \frac{2W_{j}}{\varphi^2}$.
For case 1.1, if it is applied for solution $A$, we have $a_1^j=a_1^{j-1}$, and if it is applied for solution $B$, we have $b_1^j=b_1^{j-1}$. By $OPT_j\geq OPT_{j-1}$, it remains to show the lower bounds on the completion times of the first machine for both solutions and case 1.1. If case 1.1 is applied for solution $A$, we have $$a_2^j=a_2^{j-1}+p_j\leq \frac{W_j}{\varphi^2} \ , $$ and thus $a_1^j \geq W_j-\frac{W_j}{\varphi^2}=\frac{W_j}{\varphi}$.
If case 1.1 is applied for solution $B$, we have $b_2^j=b_2^{j-1}+p_j\leq \frac{W_j}{\varphi^3}$, and thus $b_1^j \geq W_j-\frac{W_j}{\varphi^3}=\frac{2W_j}{\varphi^2}$.
\end{proof}

\medskip

We continue with the next case.

\medskip
Note that in this case we have $$\frac{W_j}2 \leq W_{j-1} \leq (\varphi-1)\cdot W_j \ , $$ by $$W_j=W_{j-1}+p_j \leq W_{j-1}+\frac{W_j}2 \mbox{ \ \ \  and \ \ \  } W_{j-1}=W_j-p_j < W_j\cdot (1-(2-\varphi)) \ . $$ We also have $p_j \leq W_{j-1}  \leq  \varphi \cdot p_j$ (by $\frac{\varphi-1}{2-\varphi}=\varphi$),  $W_j \geq \varphi \cdot W_{j-1}$ and $p_j \geq \frac{W_{j-1}}{\varphi}$.

\begin{claim}
The assignment of job $j$ in case 2 is valid in the sense that a total size of $p_j$ is assigned, and there is no overlap between the time slots allocated for $p_j$ on the two machines. No idle time is introduced, and the invariants hold after job $j$ is assigned.
\end{claim}
\begin{proof}
Similarly to case 1, the assignment does not introduce idle time. If indeed the assignment is possible and valid, the invariants will hold. We will show similar properties to those proved in case 1, taking into account the swap of the solutions.

We start with the assignments for solution $A$. To prove $a_2^{j-1} \leq \frac{W_{j}}{\varphi^3}$ (which implies that the intervals have non-negative length), we use the properties $a_2^{j-1} \leq \frac{W_{j-1}}{\varphi^2}$ and $W_{j-1} \leq (\varphi-1)W_j=\frac{W_j}{\varphi}$. To prove $a_1^{j-1} \leq \frac{2W_{j}}{\varphi^2}$, we use $a_1^{j-1} \leq W_{j-1} \leq (\varphi-1) \cdot W_j$ and $\varphi-1 < \frac{2}{\varphi^2}$.
To prove $\frac{W_j}{\varphi^3}\leq a_1^{j-1}$ (so that we can see that there is no overlap between the intervals assigned to one job on the two machines) we use the invariant $a_1^{j-1} \geq \frac{W_{j-1}}{\varphi}$, and the properties $W_{j-1} \geq  \frac{W_j}2$ and $2\cdot \varphi < \varphi^3$. The first two inequalities show that non-negative length intervals are use for $j$, and the third inequality shows that there is no overlap between the part of $j$. Since previously a total size of $W_{j-1}$ was assigned and after the assignment of $j$, a total size of $W_j$ is assigned, $j$ was scheduled completely.

Consider now the assignments for solution $B$. The reasoning regarding the assignment is the same as for $A$.
To prove $b_2^{j-1} \leq \frac{W_{j}}{\varphi^2}$, we use the property $b_2^{j-1} \leq \frac{W_{j-1}}{\varphi^3}$ and $W_{j-1} < W_j$. To prove $b_1^{j-1} \leq \frac{W_{j}}{\varphi}$, we use  $W_{j-1} \leq (\varphi-1)\cdot W_j$, and get
$b_1^{j-1} \leq W_{j-1}  \leq \frac{W_j}{\varphi}$.
To prove $\frac{W_j}{\varphi^2}\leq b_1^{j-1}$, we use the invariant $b_1^{j-1} \geq \frac{2W_{j-1}}{\varphi^2}$ and $W_{j-1} \geq  \frac{W_j}2$.
\end{proof}

\medskip

Finally, we analyze case 3.
By the properties, the loads after the fake job is assigned are based on the total size $2\cdot W_{j-1}$, which allows this assignment.

We have $p_j = W_j -W_{j-1} < 2p_j -W_{j-1}$ and therefore $p_j>W_{j-1}$ holds, that is, $j$ is indeed larger than all previously arrived jobs together.

\begin{claim}
The assignment of job $j$ in case 3 is valid in the sense that a total size of $p_j$ is assigned, and there is no overlap between the time slots allocated for $p_j$ on the two machines. No idle time is introduced, and the invariants hold after job $j$ is assigned.
\end{claim}
\begin{proof}
The assignment of the part of $j$ of size $W_{j-1}$ satisfies all requirements and the invariants hold, since it is assigned by case 2. Moreover, the machine completion times are exactly $\frac{2W_{j-1}}{\varphi}$ for the first machine and solution $A$, $\frac{2W_{j-1}}{\varphi^2}$ for the second machine and solution $A$,
$\frac{4W_{j-1}}{\varphi^2}$ for the first machine and solution $B$, and $\frac{2W_{j-1}}{\varphi^3}$ for the second machine and solution $B$. In particular, the load of the first machine is larger than the load of the second machine, for each of the two solutions.
Thus, the assignment of the remaining part of $j$ is not in parallel with any part of a job (including $j$), and no idle time is created. The added part completes the assignment of a total size of $p_j$ for $j$ in each one of the solutions.

Next, we prove the invariants. The left hand sizes of \eqref{n1} and \eqref{n2} hold since all the remaining size was added to the first machine, while $\frac{1}{\varphi}<1$ and $\frac{2}{\varphi^2}<1$ hold. We have $$a_1^j=\frac{2W_{j-1}}{\varphi}+p_j-W_{j-1}=W_{j-1}\cdot (\frac 2{\varphi}-1)+p_j = (W_j-p_j)\cdot (R-1)+p_j$$ $$=(R-1)\cdot W_j+(2-R) \cdot p_j \leq 2(R-1)\cdot OPT_j+ (2-R)\cdot OPT_j=R\cdot OPT_j \ .$$
Finally,  $$b_1^j=\frac{4W_{j-1}}{\varphi^2}+p_j-W_{j-1}=W_{j-1}\cdot (\frac 4{\varphi^2}-1)+p_j = (W_j-p_j)\cdot (R^2-1)+p_j$$ $$=(R^2-1)\cdot W_j+(2-R^2)\cdot p_j \leq 2(R^2-1)\cdot OPT_j+ (2-R^2)\cdot OPT_j=R^2\cdot OPT_j \ ,$$ as required.
\end{proof}

\medskip

Since all options for the size of $j$ are covered by the three cases, se summarize the design and analysis of our algorithm with the following theorem, where the lower bound follows from Proposition \ref{propo2}, and the upper bound follows from the first invariant.
\begin{theorem}
The algorithm above has competitive ratio $R\approx 1.236068$ for two machines and two solutions, and this is the best possible competitive ratio for this problem.
\end{theorem}

\section{Two identical machines, decreasing sizes}\label{sor}
This section contains the algorithm for inputs with non-increasing sizes. In this case, we will assume that $p_1=1$, without loss of generality.
Define $\alpha=\sqrt{6}-2\approx 0.4495$, where $\alpha^2+4\alpha=2$. Let $R=6-2\cdot \sqrt{6}\approx 1.10102$, and $$r=\frac{3(2-R)}{2}=3(\sqrt{6}-2)\approx 1.34847 \ . $$ We have $R^2=12(R-1)$ and $r=3\cdot \alpha$, which implies $r^2+12r=18$. Let $\beta=1-\frac{\sqrt{6}}3\approx 0.1835$, where $\beta=\frac{\alpha}{2+\alpha}$.
We also have $$R+r-\alpha=(6-2\cdot \sqrt{6})+(3(\sqrt{6}-2))-(\sqrt{6}-2)=2 \ . $$

Here, we also start with a lower bound on the competitive ratio.

\begin{proposition}\label{propo3}
The competitive ratio of any algorithm for $m=2$ identical machines is at least $R$.
\end{proposition}
\begin{proof}
Assume by contradiction that there is an algorithm whose competitive ratio is $\rho<R$.

We will use two inputs, where the prefix of the first two jobs is identical for the two inputs, and the third job is different. The algorithm obviously has to satisfy the definition of the competitive ratio for all inputs, including these two inputs. Moreover, the prefix of the first two jobs is treated as an input as well.

The inputs start with two jobs of size $1$ each. For the input consisting of these two jobs, an optimal offline solution assigns one job to each machine, without preemption or idle time. We use the same notation as in  Lemma \ref{propo2}, which is the other lower bound proof (the variables used for the solutions, after the sand jobs were presented). Since the total size of the two jobs here is $2$, we have $t_1+t_2=x_1+x_2=2$. We also assume again that $t_1 \leq x_1$  holds.

Since the makespan of the algorithm for the input consisting of two jobs of size $1$ each, based on the first solution, is at least $t_1$, we have $t_1\leq \rho$ by the definition of competitive ratio for this input. We find that $t_2=2-t_1\geq 2-\rho >2-R$. For the two other inputs, there is a third job for each input. Its size is $1$ for the first input, and it is $\alpha$ for the second input. For the input where the size of the third job is $1$, the makespan is at least $\min\{t_2+1,x_2+1\}=x_2+1$. An optimal offline solution has makespan $\frac 32$, and therefore $x_2+1 \leq \rho \cdot 1.5$. Since $x_1=2-x_2$, we find that $3-x_1 \leq 1.5 \cdot \rho$.
Thus, by $\rho<R$, we have $$x_1 \geq 3- 1.5 \cdot \rho > 3- 1.5 (6-2\sqrt{6})=3\sqrt{6}-6 = r \ . $$
For the input where the size of the second job is $\alpha$, an optimal solution has cost $\frac{2+\alpha}2$. The algorithm has makespan of at least $\min\{t_2+\alpha,x_1\}$. By $t_2 > 2-R$, we have $t_2+\alpha  >2-R+\alpha = r$. We find that $\min\{t_2+\alpha,x_1\}>r$.
The competitive ratio for this input is therefore above $\frac{2r}{2+\alpha}=\frac{6(\sqrt{6}-2)}{\sqrt{6}}=6-2\sqrt{6}=R$, a contradiction to the assumption on the competitive ratio for the algorithm.
\end{proof}

Next, we design an algorithm for inputs consisting of jobs arriving in a sorted order (sorted by non-increasing sizes). The algorithm is designed for two identical machines and constructs two solutions.
We use the same notation as in the previous section.
The assignment of the first job (of size $1$) is always during the time interval $[0,1)$ on the first machine for both solutions. Thus, $a_1^1=b_1^1=1$, and $a_2^0=b_2^0=0$. Moreover, $W_1=1$. These solutions are both optimal, and the analysis of the algorithm will start after the assignment of the second job.

For this variant, we sometimes use only one solution starting with a specified input job. The meaning is that we consider   one  specific solution towards the competitive ratio, and in the other solution the assignment is arbitrary. We will define the time in which we will stop considering the second solution  precisely.

The algorithm uses one of two approaches for designing the solutions. If the second job is relatively small, all future jobs are small as well due to the sorted order of the input. In this case only one solution is used already starting with the second job. The algorithm assigns the first job to the first machine. Then, it assigns jobs or parts of jobs to the second machine, until the second machine receives a sufficient total size, and the required balance is satisfied. From that time on, jobs are assigned so as to keep a fixed load ratio. In the second approach, both solutions are used.
The main idea for each solution is as follows. First, as long as the total size of jobs is relatively small (compared to the size of the first job, which is assumed to be equal to $1$), only one of the machines is used. However, since the load ratios are different, the required load to be achieved is not the same for the two solutions. In this approach, the solutions are swapped if the new job is not small, as in the previous algorithm. There are multiple cases since the situation may be different with respect to the two solutions. However, due to the sorted order, we do not have the situation with a large job, and the most significant assignments are those of the first few jobs.

In the second approach, there is also a case where the second job is small. In this case one can also stop using the second solution after several jobs have arrived (and thus new jobs, which are not larger than previous jobs, are relatively small compared to the total size of already existing jobs). Case 1 of the second apporach deals, in particular, with assigning small jobs. Since job sizes are non-increasing while the function $W_j$ is monotonically increasing, once some job is assigned by case 1 due to its small size, all further jobs will be assigned in the same way (they are assigned using case 1, and due to the same reason of being small). However, case 1 is applied in another situation as well, and therefore we will not take this into account, and we define and analyze both solutions for the first case of the second approach.

\medskip

\noindent{\bf The algorithm.}

The first job is assigned into the time interval $[0,1)$ for both solutions.
Next, for the assignment of further jobs, the algorithm proceeds in one of two approaches, based on the size of the second job, $p_2$.
The first approach is used if $p_2 \leq 0.4$, and otherwise the second approach is used.

\noindent{\it The first approach.}
In this approach, only one solution will be used starting from the second job.
The assignment of job $j$ (for $j\geq 2$) is as follows. If $$a_2^{j-1}+p_j \leq (1-\frac R2)\cdot W_j \ , $$ job $j$ is assigned to the second machine completely during $[a_2^{j-1},a_2^{j-1}+p_j)$, and this is the first case of the first approach. Otherwise, job $j$ is assigned to the second machine during $[a_2^{j-1},(1-\frac R2)\cdot W_j)$, and to the first machine during $[a_1^{j-1}, \frac R2\cdot W_j)$, and this is the second case of the first approach.

\medskip
\noindent{\it The second approach.}
The second approach is for the case $p_2>0.4$.

In this approach, we define the assignment of the second job separately. For solution $A$, the second job is assigned during $[1,R)$ on the first machine and during $[0,p_2+1-R)$ on the second machine. For solution $B$, the second job is assigned during $[1,r)$ on the first machine and during $[0,p_2+1-r)$ on the second machine.

The remaining cases are applied for any job $j\geq 3$.

\medskip
\noindent{\bf Case 1.} In this case, we assume that at least one of the following two conditions hold: $p_j \leq W_j \cdot (1-\frac{\sqrt{6}}3)$, where $1-\frac{\sqrt{6}}3 \approx 0.183503$ and $W_j\leq \sqrt{6}=\frac{6-R}{2}\approx 2.4494897$ (it is possible that both conditions hold simultaneously, but we apply this case even if just one of the two conditions holds).
In this case the solutions are not swapped, and we describe the assignment for each one of the solutions separately.

\smallskip

{\bf Case 1.1 for solution $\boldsymbol{A}$.} In this case we assume that we additionally have $a_2^{j-1}+p_j \leq \frac{2-R}2 \cdot {W_j}$.
Job $j$ is assigned to the second machine completely during $[a_2^{j-1},a_2^{j-1}+p_j)$.

\medskip

{\bf Case 1.2 for solution $\boldsymbol{A}$.}
If case 1.1 was not applied, job $j$ is assigned to the second machine during $[a_2^{j-1}, {W_j} \cdot {\frac{2-R}2})$, and to the first machine during $[a_1^{j-1}, {W_j}\cdot{\frac R2})$.

\medskip

{\bf Case 1.1 for solution $\boldsymbol{B}$.} In this case we assume that we additionally have $b_2^{j-1}+p_j \leq \frac 25 \cdot {W_j}$.
Job $j$ is assigned to the second machine during $[b_2^{j-1},b_2^{j-1}+p_j)$.

\medskip

{\bf Case 1.2 for solution $\boldsymbol{B}$.}
If case 1.1 was not applied, job $j$ is assigned to the second machine during $[b_2^{j-1},\frac 25 \cdot {W_j})$, and to the first machine during $[b_1^{j-1},\frac 35 \cdot {W_j})$.

\medskip

\noindent{\bf Case 2.} 
In this case, which is applied in any remaining situation, the solutions are swapped. We describe the assignment for each one of the solutions.

For solution $A$ (which becomes solution $B$), $j$ is assigned as follows. Let $$\Gamma_j=\min\{a_2^{j-1}+p_j,a_1^{j-1},0.4\cdot W_j\} \ .$$
Job $j$ is assigned to the second machine during the time interval $[a_2^{j-1},\Gamma_j)$, and it is assigned to the first machine during the time interval $[a_1^{j-1},W_j-\Gamma_j)$.

For solution $B$ (which becomes solution $A$), $j$ is assigned to the second machine, during the time interval $[b_2^{j-1},\frac{2-R}{2}\cdot W_j)$, and to the first machine during the time interval $[b_1^{j-1},\frac{R}{2}\cdot W_j)$.

\medskip

\noindent{\bf Analysis.}

Next, we analyze the algorithm.
The assignment of the first job is optimal since the load is equal to $p_1$ (for both solutions). Thus, we will deal with the case $j>1$.

We start with proving the competitive ratio for the first approach.
For solution $A$, we will require that $$\max\{1,\frac{R}2 \cdot {W_j}\} \leq a_1^j \leq  R \cdot OPT_j$$ will always hold. By this constraint which will be proved as an invariant, and by the lack of idle time, we will have $a_2^j \leq (1-\frac R2)\cdot W_j$ (which can also be seen directly from the assignment in both cases). In particular, this means (by $R>1$) that $a_1^j > a_2^j$. For $j=1$, we have $a_1^j=W_j=1$, $OPT_j=1$ and the invariant holds by $1<R<2$. We will always have $a_j^1 \geq a_1^1= 1$ (for any $j\geq 1$), and therefore we will prove that  $\frac{R}2 \cdot {W_j} \leq a_1^j \leq  R \cdot OPT_j$ holds after every additional job assignment (that is, for $j\geq 2$).
In both cases of the first approach no idle time is introduced. We will consider the further details for each case separately.

In the first case, the new job is assigned completely to one machine, and there cannot be an overlap. Additionally,  in this case we have $a_1^j = a_1^{j-1} \leq  R \cdot OPT_{j-1} \leq  R \cdot OPT_{j}$, $a_2^j \leq (1-\frac R2)\cdot W_j$, and therefore, since $a_1^j+a_2^j=W_j$, we have $a_1^j \geq \frac R2\cdot W_j$, so the invariant holds.

In the second case the job is assigned completely since after the assignment the total assigned size is $W_j$. We show that there is no overlap between the two parts of the job, that is, $$(1-\frac R2)\cdot W_j \leq a_1^{j-1} \ , $$  If
$(1-\frac R2)\cdot W_j \leq 1$, we are done by $a_1^{j-1} \geq 1$, and therefore are left with the case $(1-\frac R2)\cdot W_j >1$, or alternatively, $W_j > \frac{2}{2-R}$. By $a_1^{j-1} \geq \frac{R}2 \cdot {W_{j-1}}$, it is sufficient to prove that $(1-\frac R2)\cdot W_j \leq \frac{R}2 \cdot {W_{j-1}}$ holds. By $W_j=W_{j-1}+p_j$, this inequality is equivalent to $$(1-\frac R2)\cdot W_j \leq \frac{R}2 \cdot {W_{j}}-\frac R2 \cdot p_j \ , $$ or alternatively,
$R \cdot p_j \leq 2(R-1) \cdot {W_{j}}$.
We have $\frac{2(R-1)}R \cdot {W_{j}}  > \frac{2(R-1)}R \cdot \frac{2}{2-R}$, and we are done by $\frac{2(R-1)}R \cdot \frac{2}{2-R} > 0.408$ and $p_j \leq 0.4$.
It is left to show that the intervals have non-negative lengths in the second case, that is, $a_1^{j-1} \leq \frac R2\cdot W_j$ and $a_2^{j-1} \leq (1-\frac R2)\cdot W_j$. Since $$W_j=a_1^{j-1}+a_2^{j-1}+p_j \mbox{ \ \ \ \ and \ \ \ \ } a_2^{j-1}+p_j > (1-\frac R2)\cdot W_j$$ hold, we have $a_1^{j-1} < W_j-(1-\frac R2)\cdot W_j =\frac R2 \cdot W_j$. By the invariant $a_1^{j-1} \geq \frac R2 \cdot W_{j-1}$, we have $a_2^{j-1} \leq \frac{2-R}2 \cdot W_{j-1} < \frac{2-R}2 \cdot W_{j}$.
The invariant will hold for the second case since $a_1^j=\frac R2 \cdot W_j$, and $OPT_j \geq \frac{W_j}2$.

\medskip

Next, we proceed with proving the analysis of the competitive ratio for second approach.
Recall that the case $j=2$ was different. We prove this case separately.
The interval on the second machine is non-empty since $R<r<1.35$ and $$p_2+1-R > p_2+1-r > 0.4+1- 1.35 >0 \ . $$
The second job is assigned without idle time, and without overlap since $p_2+1-r < p_2+1-R <1$ by $p_2\leq p_1=1$ and $R>1$. Thus, $a_1^2=R$ and $b_1^2=r$.
Note that $OPT_2=1$, which implies $a_1^2=R\cdot OPT_2$ (so the competitive ratio is not violated) and $b_1^2= r \cdot OPT_2$. This allows us to define the following invariants for $j\geq 2$: $$\max\{R,\frac{R}2 \cdot {W_j}\} \leq a_1^j \leq  R \cdot OPT_j \mbox{\ \  \ \ and  \ \ \ } \max\{r, 0.6 \cdot {W_j}\} \leq b_1^j \leq  r \cdot OPT_j \ . $$  By these invariants and lack of idle time, we will always have $a_2^j \leq \frac{2-R}2 \cdot W_j$ and $b_2^j \leq 0.4 \cdot W_j$ (and the competitive ratio will hold using solution $A$).

For $j=2$, we have $W_j=p_1+p_2 \leq 2\cdot p_1=2$, so $\frac{R}2 \cdot {W_j} \leq R$ and $0.6 \cdot {W_j} \leq 1.2 <r$. The invariants hold for $j=2$ since $a_1^2=R=R\cdot OPT_j$ and $a_1^2=r=r\cdot OPT_j$.
Since the loads cannot decrease over time, and by the assignment of the second job, for $j\geq 2$ it always holds that $a_1^j \geq R$ and $b_1^j \geq r$, and thus, it will be sufficient to prove $\frac{R}2 \cdot {W_j} \leq a_1^j \leq  R \cdot OPT_j$ and $0.6 \cdot W_j \leq b_1^j \leq  r \cdot OPT_j$ for $j\geq 3$ (using induction and the assignment of each job). There are several cases for the assignment. Since the cases $j=1,2$ were analyzed completely, in what follows, we always assume that $j\geq 3$. Note that for $j\geq 3$, it holds that $W_j \geq p_1+p_2+p_j \geq 3\cdot p_j$, and therefore $p_j \leq \frac {W_j}3$.

\begin{claim}
The assignment of job $j$ in case 1 is valid in the sense that a total size of $p_j$ is assigned, and there is no overlap between the time slots allocated for $p_j$ on the two machines. No idle time is introduced, and the invariants hold after job $j$ is assigned.
\end{claim}
\begin{proof}
For case 1.1 and both solutions, the assignment cannot be invalid since the job is not preempted. No idle time is introduced since the machine that receives the job is the second machine, and the interval allocated for the job starts at time $a_2^{j-1}$ for the first solution and $b_2^{j-1}$ for the second solution, and its length is equal to the size of the assigned job. In fact, for case 1.2, while the job is preempted and split into two parts, each part is assigned without introducing any idle time.

In order to prove all properties excluding the invariants, it is left to analyze the validity of assignment in the remaining cases and to show that the job is fully assigned in these cases. We show that the two parts of $j$ do not overlap for both solutions, the parts have non-negative sizes, and their total size is equal to $j$.

We start with the assignments for solution $A$.
We show that in case 1.2  the intervals have positive lengths. We have $$a_2^{j-1}\leq \frac {2-R}2 \cdot {W_{j-1}}<\frac {2-R}2 \cdot {W_{j}}\ , $$ by the invariants for the previous time $j-1$, and since $p_j>0$. Additionally,  since case 1.1 not applied, it holds that $a_2^{j-1}+p_j > \frac {2-R}2 \cdot {W_{j}}$, which implies that $$a_1^{j-1}=W_j-p_j-a_2^{j-1} <W_j-\frac {2-R}2 \cdot {W_{j}}=\frac R2 \cdot {W_j} \ , $$ and thus $a_1^{j-1} \leq\frac R2 \cdot {W_j} $ holds. We have proved that both lengths of the allocated intervals are indeed positive.
Moreover, $$(\frac{2-R}2 \cdot {W_{j}}-a_2^{j-1})+(\frac{R}2 \cdot {W_{j}}-a_1^{j-1})=W_j-W_{j-1}=p_j \ , $$ and $j$ is assigned completely.

To show that there is no overlap between the parts, we prove that $\frac {2-R}2 \cdot {W_{j}} \leq  a_1^{j-1}$ holds.  We start with the case $p_j \leq W_j \cdot (1-\frac{\sqrt{6}}3)$.
We have
$W_j-W_{j-1}=p_j \leq W_j \cdot (1-\frac{\sqrt{6}}3)$, and therefore $W_{j-1} \geq W_j \cdot \frac{\sqrt{6}}3$ (where $\frac{\sqrt{6}}3 \approx 0.81649658$). Using the invariant for the loads, we have $a_1^{j-1} \geq \frac R2 \cdot W_{j-1}$, and we find  that $a_1^{j-1} \geq  \frac R2 \cdot W_j \cdot \frac{\sqrt{6}}3$. It is left to see that $\frac{R}{\sqrt{6}} \geq\frac {2-R}2$ holds, which is equivalent to $R\cdot (\frac{1}{2}+\frac{1}{\sqrt{6}}) \geq 1$, and indeed $$R\cdot(\frac{1}{2}+\frac{1}{\sqrt{6}})=(6-2\cdot \sqrt{6})\cdot(\frac{3+\sqrt{6}}{6})=1 \ . $$
In the case $W_j\leq \frac{6-R}{2}$, we find $\frac {2-R}2 \cdot {W_{j}}\leq \frac{(2-R)\cdot(6-R)}{4}=R$, and we are done by $a_1^{j-1} \geq R$.

We continue with the assignments for solution $B$.
We show that in case 1.2  the intervals have positive lengths. We have $b_2^{j-1}\leq 0.4 \cdot {W_{j-1}}< 0.4\cdot {W_{j}}$, by the invariants for the previous time $j-1$, and since $p_j>0$. Additionally,  since case 1.1 not applied, it holds that $$b_2^{j-1}+p_j > 0.4 \cdot {W_{j}} \ , $$ which implies $$b_1^{j-1}=W_j-p_j-b_2^{j-1} <W_j- 0.4 \cdot {W_{j}}=0.6 \cdot {W_j} \ , $$ and thus $b_1^{j-1} \leq 0.6 \cdot {W_j} $ holds.
Moreover, $(0.4 \cdot {W_{j}}-b_2^{j-1})+(0.6 \cdot {W_{j}}-b_1^{j-1})=W_j-W_{j-1}=p_j$, and $j$ is assigned completely.
To show that there is no overlap between the parts, we prove that $0.4 \cdot {W_{j}} \leq  b_1^{j-1}$ holds.  We start with the case $p_j \leq W_j \cdot (1-\frac{\sqrt{6}}3)$.
Recall that we have $W_{j-1} \geq W_j \cdot \frac{\sqrt{6}}3$. Using the invariant for the loads, we have $b_1^{j-1} \geq 0.6 \cdot W_{j-1}$, and we find  that $b_1^{j-1} \geq  0.6 \cdot W_j \cdot \frac{\sqrt{6}}3 > 0.4 \cdot W_j$, since $0.6 \cdot \frac{\sqrt{6}}3 \approx 0.4898979485$. In the case $W_j\leq \frac{6-R}{2}$, by  $$0.4 \cdot W_j < \frac {2-R}2 \cdot {W_{j}} \leq  \frac{(2-R)\cdot(6-R)}{4}= R \ , $$ we are done by $b_1^{j-1} \geq r >R$.

We now show that the invariants will hold in all cases after the assignment. For case 1.2, if it is applied for solution $A$, the exact value of $a_1^j$ is $R \cdot \frac{W_{j}}{2}$. Similarly, if it is applied for solution $B$, the exact value of $b_1^j$ is $0.6 \cdot W_j$. This proves that the left hand side of the two constraints holds. Since $OPT_j \geq \frac{W_2}2$, and $r>1.2$, the right hand side of the constraints holds as well.

For case 1.1, if it is applied for solution $A$, we have $a_1^j=a_1^{j-1}$, and if it is applied for solution $B$, we have $b_1^j=b_1^{j-1}$. By $OPT_j\geq OPT_{j-1}$, it remains to show the lower bounds on the completion times of the first machine for both solutions and case 1.1. If case 1.1 is applied for solution $A$, we have $a_2^j=a_2^{j-1}+p_j\leq \frac{2-R}2 \cdot W_j$, and thus $a_1^j \geq W_j-\frac{2-R}2 \cdot W_j=\frac R2 \cdot W_j $.
If case 1.1 is applied for solution $B$, we have $b_2^j=b_2^{j-1}+p_j \leq 0.4\cdot W_j$, and thus $b_1^j \geq W_j-0.4\cdot W_j = 0.6 \cdot W_j $.
\end{proof}

We completed the analysis for case 1, and continue with case 2.
In this remaining case we have $p_j > W_j \cdot (1-\frac{\sqrt{6}}3)$, and $W_j >\frac{6-R}{2}=\sqrt{6}$. This yields $$W_j-W_{j-1}=p_j> W_j \cdot (1-\frac{\sqrt{6}}3) \mbox{ \ \ \ and \ \ \ } W_{j-1}<\frac{\sqrt{6}}3 \cdot W_j \ . $$ By $W_j >\sqrt{6}>2$, it holds that $OPT_j=\frac{W_j}2>1$. The required constraints of the invariants will therefore be $a_1^j=R \cdot \frac{W_j}2$ (since the two bounds from above and from below are equal), and $0.6 \cdot W_j \leq b_1^j \leq \frac r2 \cdot W_j$.
 Note that $$\Gamma_j=\min\{a_2^{j-1}+p_j,a_1^{j-1},0.4\cdot W_j\}$$ is strictly positive since $j\geq 3$, and all three considered values are positive. However, none of the values can exceed $W_j$, so this value is not larger than $W_j$.
\begin{claim}
The assignment of job $j$ in case 2 is valid in the sense that a total size of $p_j$ is assigned, and there is no overlap between the time slots allocated for $p_j$ on the two machines. No idle time is introduced, and the invariants hold after job $j$ is assigned.
\end{claim}
\begin{proof}
By the assignment, no idle time is introduced. Thus, the solutions will not have idle time (since they are valid, which is proved here). Moreover, if the solutions are valid, $j$ is assigned completely. This holds for the two solutions because there is no idle time and the total load after the assignment is $W_j$.
It is left to show that the lengths of the intervals are non-negative, there is no overlap between the parts, and the invariants will hold after the assignment.

Now, we show that the lengths of the intervals are non-negative, and that there is no overlap between the parts of the jobs. The six properties to be shown are
$b_1^{j-1} \leq \frac R2 \cdot W_j$, $a_1^{j-1} \leq W_j-\Gamma_j$, $b_2^{j-1} \leq \frac {2-R}2 \cdot W_j$, $a_2^{j-1} \leq \Gamma_j$, $\frac{2-R}2 \cdot W_j \leq b_1^{j-1}$ and $\Gamma_j \leq a_1^{j-1}$. The sixth condition holds directly by definition of $\Gamma_j$. Similarly, the second condition also holds by this definition since $$\Gamma_j \leq a_2^{j-1}+p_j= W_j-a_1^{j-1} \ . $$ For the fourth condition, we consider all possibilities for $\Gamma_j$. If $\Gamma_j=a_2^{j-1}+p_j$, we are done since $p_j>0$. If $\Gamma_j=a_1^{j-1}$, we are done by $$a_2^{j-1} \leq \frac{2-R}{2}\cdot W_{j-1}< \frac{R}{2}\cdot W_{j-1} \leq a_1^{j-1} \ . $$ If $\Gamma_j=0.4\cdot W_j$, we are done by $$a_2^{j-1}\leq  \frac{2-R}{2}\cdot W_{j-1} <\frac{2-R}{2}\cdot  \frac{\sqrt{6}}3 \cdot W_j < 0.368\cdot W_j < 0.4 \cdot W_j \ . $$

The properties $b_1^{j-1} \geq 0.6 \cdot W_{j-1}$ and $b_2^{j-1} \leq 0.4 \cdot W_{j-1}$ follow from the invariants for $j-1$. Moreover, by $OPT_{j-1} = \max\{1,\frac{W_{j-1}}2\}$, it holds that $$b_1^{j-1} \leq r \cdot \max\{1,\frac{W_{j-1}}2\} \ . $$

Now, we prove that $b_2^{j-1} \leq \frac {2-R}2 \cdot W_j$ holds.
We use $b_2^{j-1} \leq 0.4 \cdot W_{j-1}$, and it is sufficient to prove that
$0.4 \cdot W_{j-1}\leq \frac{2-R}{2}\cdot W_j$ holds. By $W_{j-1}<  W_j$, it is sufficient to show that $0.4 \leq \frac{2-R}{2}$ holds, which is indeed the case as $\frac{2-R}{2} >0.449$.

Next, we prove that $b_1^{j-1} \leq \frac R2 \cdot W_j$ holds.
We use $W_j > \sqrt{6}$ for the situation $b_1^{j-1} \leq r$, and we find that $\frac{R}{2}\cdot W_j > \frac{R}{2} \cdot \sqrt{6}=r$. Otherwise, $b_1^{j-1} > r$, and by the upper bound in the invariant for $b_1^{j-1}$ we get $b_1^{j-1} \leq r \cdot \frac{W_{j-1}}2$ and $W_{j-1}>2$. By $W_{j-1}<\frac{\sqrt{6}}3 \cdot W_j$, we get $$b_1^{j-1} \leq r \cdot \frac{(\frac{\sqrt{6}}3) \cdot W_j}2 = \frac{R}2 \cdot W_j \ . $$

Now, we prove that $\frac{2-R}2 \cdot W_j \leq b_1^{j-1}$ holds.
We use $b_1^{j-1} \geq 0.6 \cdot W_{j-1}$, and it is sufficient to prove that $0.6 \cdot W_{j-1} \geq \frac{2-R}{2}\cdot W_j$.
If $j\geq 4$, we have $W_j=W_{j-1}+p_j$, where $$W_{j-1} \geq (j-1)\cdot p_j = (j-1) \cdot (W_j-W_{j-1}) \ , $$  and therefore $W_j \leq \frac{j}{j-1} \cdot W_{j-1} \leq \frac 43 \cdot W_{j-1}$. Thus, it is sufficient to prove that $\frac{2-R}2 \cdot \frac 43  \leq 0.6$. Indeed, it holds that $\frac{2-R}2 \cdot \frac 43  < 0.59932$. In the case $j=3$, we have $b_1^{j-1}= r$, thus, it is required to prove $\frac{2-R}{2}\cdot W_j \leq r$. Since $W_3 \leq 3$, it is sufficient to prove
$3 \cdot \frac{2-R}{2} \leq r$. This holds with equality by the definition of $r$.

Finally, we show that the constraints will hold. For solution $A$ which was obtained from solution $B$, we have $a_1^j=\frac{R}{2}\cdot W_j$, which proves that the invariants hold since $OPT_j  =\frac{W_j}2$. For solution $B$ which was obtained from solution $A$, we will prove $ 0.6 \cdot W_j \leq b_1^j \leq \frac r2 \cdot W_j$. Since $\Gamma_j \leq 0.4 \cdot W_j$, we have $$b_1^j=W_j-\Gamma_j \geq 0.6 \cdot W_j \ .$$

It remains to prove that $W_j-\Gamma_j = b_1^j \leq \frac r2 \cdot W_j$ holds.
We consider the three cases for $\Gamma_j$. If $\Gamma_j=a_2^{j-1}+p_j$, we have $b_1^j=a_1^{j-1} \leq R \cdot OPT_{j-1}$. We have $OPT_{j-1}=\max\{1,\frac{W_{j-1}}2\} \leq \max\{1,\frac{W_{j}}2\}=\frac{W_j}2$. Thus, $b_1^j \leq \frac{R}{2}  \cdot W_{j} < \frac r2 \cdot W_j$.

Otherwise, if $\Gamma_j=a_1^{j-1}$, we have $b_1^j=W_j-a_1^{j-1}=a_2^{j-1}+p_j$, where $a_2^{j-1}\leq \frac{2-R}2 \cdot W_{j-1}$, and therefore $$b_1^j \leq \frac{2-R}2 \cdot W_{j-1} +p_j=\frac{2-R}2\cdot W_j + \frac R2 \cdot p_j \ .$$ We have $p_j \leq \frac{W_j}3$, since $W_j \geq p_1+p_2+p_j\geq 3\cdot p_j$. Thus, $$b_1^j \leq W_j\cdot (1-\frac R2+ \frac R6)=(1-\frac R3) \cdot W_j \ .$$ We are done since $\frac r2  >0.67$ and $1-\frac R3 <0.64$.

In the last case, $\Gamma_j=0.4\cdot W_j$, we are done since $W_j-\Gamma_j = 0.6 \cdot W_j \leq \frac r2 \cdot W_j$, since $r>1.2$.
\end{proof}

Our two approaches and the cases into which the algorithm was split cover all possible cases. Thus, we summarize the design and analysis of our algorithm with the following theorem, where the lower bound follows from Proposition \ref{propo3}, and the upper bound follows from the invariants.
\begin{theorem}
The algorithm above has competitive ratio $R\approx 1.10102051$ for two machines and two solutions, for inputs with non-increasing job sizes, and this is the best possible competitive ratio for this problem.
\end{theorem}

\section{Conclusion}
One can study various online problems with respect to a small number of parallel solutions. The most natural generalizations of the current work are as follows. For $m>2$ multiple machines, one goal can be to design an algorithm that uses two solutions, such that this algorithm has a smaller competitive ratio compared to earlier work  \cite{ChVlWo95,Sg97}. That is, the goal is to design an algorithm with two solutions for the problem that was studied here (with unsorted inputs) for a larger number of machines, whose competitive ratio is smaller than $\frac{e}{e-1}$ for any $m$, or to design an algorithm whose competitive ratio is smaller than $\frac{m^m}{m^m-(m-1)^m}$ for a fixed value of $m$ (for example, an algorithm of competitive ratio below $\frac{27}{19}\approx 1.42105$ for $m=3$).  Other generalizations can be an analysis for uniformly related machines, or a study of all these problems for a different small number of solutions. Sorted inputs with $m>2$ identical or uniformly related machines and multiple solutions can be studied as well.

\bibliographystyle{abbrv}


\end{document}